\documentclass[journal,onecolumn, 12pt]{IEEEtran}
\usepackage{cite}
\usepackage{amsmath,amssymb,amsfonts}
\usepackage{graphicx}
\usepackage{textcomp}
\usepackage[margin=2cm]{geometry}
\usepackage{graphics} 
\usepackage{epsfig} 
\usepackage{mathptmx} 
\usepackage{times} 
\usepackage{amsmath} 
\usepackage{amssymb}  
\usepackage{cite}
\usepackage{dirtytalk}
\usepackage{float}
\usepackage{longtable}
\usepackage{caption}
\usepackage{supertabular,booktabs}
\usepackage{blindtext}
\usepackage{multicol}
\usepackage{cellspace}%
\usepackage[english]{babel}
\usepackage{amsthm}

\usepackage{anyfontsize}
\usepackage{makecell}
\setcellgapes{3pt}
\usepackage{hyperref}
\hypersetup{hidelinks}
\hypersetup{
    colorlinks,
    linkcolor={blue},
    citecolor={blue},
    urlcolor={blue},
    }

\theoremstyle{plain}
\newtheorem{prop}{Proposition}

\theoremstyle{remark}
\newtheorem{rem}{Remark}

\newcommand\numberthis{\addtocounter{equation}{1}\tag{\theequation}}
\usepackage{tikz}
\usetikzlibrary{shapes.geometric, arrows, positioning}
\usepackage[linesnumbered,ruled,vlined]{algorithm2e}
\SetKw{Try}{try}
\SetKw{Catch}{catch}
\SetKw{Throw}{throw}
\SetKw{Return}{return}
\SetKw{Break}{break}

\SetKwComment{Comment}{$\triangleright$\ }{}

\begin{document}
%
\title{Block Diagram Analysis of an Amplitude-Frequency Design Principle in Biological Oscillators}
%
%
%

\author{Sidhanta Mohanty, \IEEEmembership{} Shaunak Sen \IEEEmembership{}
\thanks{This work was supported by the ANRF MATRICS program MTR/2023/000745.}
\thanks{$^1$Sidhanta Mohanty and $^2$Shaunak Sen are with the Department of Electrical Engineering, Indian Institute of Technology,
        New Delhi, India.
        \{$^1$\tt\small sidhanta.mohanty, $^2$\tt\small shaunak.sen\}@ee.iitd.ac.in}
}

\maketitle

\begin{abstract}
An important design principle for biological oscillators divides oscillators into two types: fixed frequency, variable amplitude and fixed amplitude, variable frequency. 
Because of the interplay of nonlinearity and feedback, both positive and negative, analytical investigations of this design principle are primarily based on numerical simulations of ordinary differential equations. 
To enhance the qualitative and quantitative characterization of this design principle, we adapted and developed a block diagram modeling framework. 
We showed how the observed amplitude-frequency characteristics could be obtained from the block diagram models of both the types of oscillators. 
We obtained explicit constraints on the positive feedback and negative feedback strengths for the oscillations to exist. 
We represented both types of oscillators in a single unifying, control-theoretic framework that connects both the types of oscillators under a single theoretical model.
These results contribute to a systems and control perspective on oscillations in biology and related contexts.
\end{abstract}


%
\IEEEpeerreviewmaketitle

\section{Introduction}
\IEEEPARstart{O}{scillatory} behaviour in biological systems is at the core of fundamental processes of life itself, such as in the cell cycle, in the regulation of circadian rhythms and in neuroscience \cite{forger2024biological}.
Therefore, a large number of studies, both phenomenological and mechanism-oriented, have been devoted to understand oscillatory behaviour and to design clocks. 
These studies have shown that the oscillation mechanisms underlying biological clocks are constituted from a dynamic interplay of nonlinear positive and negative feedback \cite{del2015biomolecular, alon2019introduction}. 

Biological oscillators have been classified, using numerical simulations, into two types, depending on whether their amplitude can be tuned for a fixed frequency or their frequency can be tuned for a fixed amplitude \cite{tsai2008robust}.
This is biologically meaningful as some oscillators, like in the human heart, may need to have the same amplitude for a wide array of frequencies, while some oscillators need to be tuned to particular frequency patterns, like the day-night cycle without major amplitude constraints. 
The oscillators with negative feedback were shown to be effectively operative at a single frequency. 
Different amplitudes near this frequency were possible but a change in frequency seemed to lead to a disappearance of oscillations. 
In contrast to this, the oscillators with both positive and negative feedback effectively operate at a fixed amplitude. 
A wide range of frequencies could be accessed at this amplitude. 
A large set of oscillators were cataloged onto this amplitude-frequency space. 
Examples of the former class of oscillators include the Goodwin oscillator \cite{goodwin1965oscillatory} and the repressilator \cite{elowitz2000synthetic}. 
Examples of the latter class include relaxation oscillators such as in the cell cycle.
The primary workflow in this study were numerical simulation of nonlinear ordinary differential equation models of various biological oscillators. 
In general, the dependence of the amplitude and the frequency on parameters is determined using simulations rather than from theoretical expressions.

\begin{figure}[H]
    \begin{center}
    \includegraphics[scale=0.4]{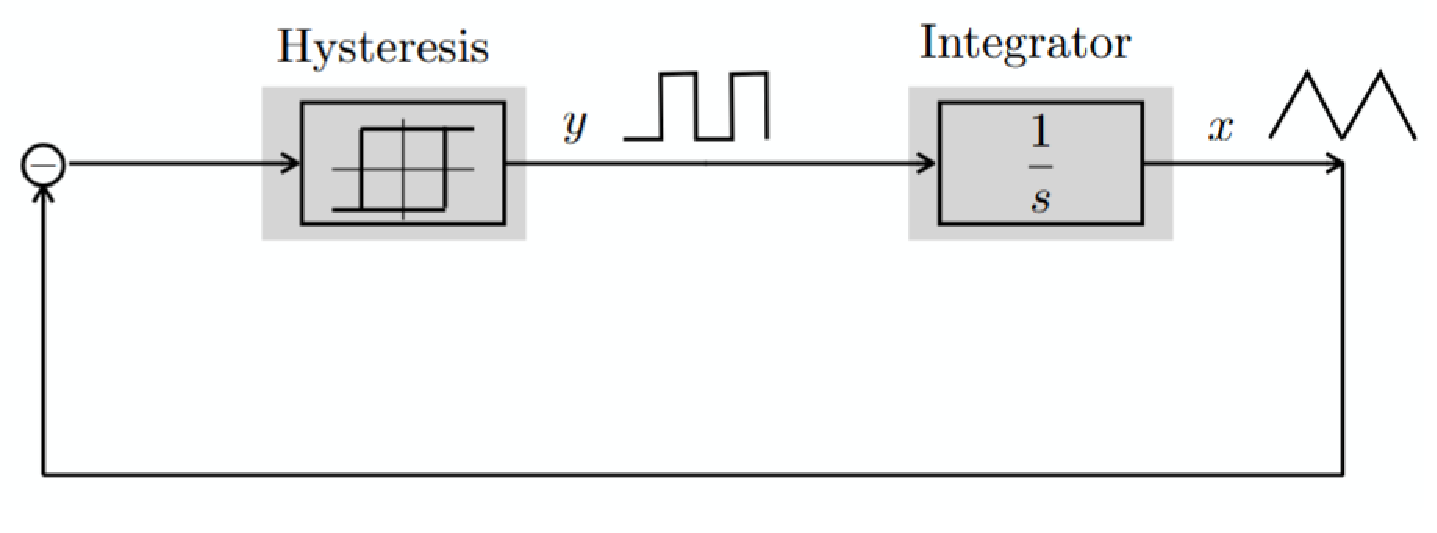}
    \caption{A hysteresis element with an integrator block in negative feedback can generate a square wave at y and a triangular wave at $x$. Based on Chapter 4 in \cite{milsum1968positive}.}
    \label{Fig1}
    \end{center}
\end{figure}
Block diagrams have been used to design oscillators in engineering (Figure \ref{Fig1}, see Chapter 4 in \cite{milsum1968positive} for a classical reference). 
Such block diagram representations form the cornerstone of introductory and practical control engineering applications.
Although the block diagram representations do not capture the full nonlinear complexity of realistic oscillator models, they provide a useful framework for analyzing their dynamical organization. 
By separating the processes into functional blocks and representing their interactions through feedback paths, block diagrams make the system structure more transparent and enable both qualitative as well as quantitative analysis of system behavior \cite{rashid2025formalization}. 
These features make block diagrams useful for examining the feedback architectures that arise in biological circuits.
Block diagram versions with negative feedback have been used to understand the systems biology of adaptation (for example in \cite{del2015biomolecular, yi2000robust}) and oscillations (for example in \cite{blanchini2018homogeneous}).
Block diagram versions with positive feedback have been used to understand hysteretic behaviour, even in systems biology contexts (for recent examples, see \cite{astrom2008feedback, kim2012regulatory, sepulchre2019control, kumar2026quantitative} and references therein). 
Additionally, there is important theoretical work focused on the analysis of existence of oscillations in block diagram models (for example in  \cite{angeli2008oscillations, miranda2018analysis, che2023dominant, che2021shaping, che2021tunable, che2022analog}). 
Therefore, block diagram representations may provide qualitative and quantitative insight into the above design principle related to the amplitude-frequency profiles of biological oscillators, which to the best of our understanding has not been reported previously.

We asked whether and how block diagram representations could be used to derive the design principle of the two types of amplitude-frequency profiles in biological oscillators. 
We showed that the fixed amplitude, variable frequency profile can arise from a saturation function embedded in an inner positive feedback loop together with an outer negative feedback loop with linear dynamics. 
We derived explicit values of the amplitude and frequency as well as the positive and negative feedback strengths required for such behaviour. 
Additionally, a sensitivity analysis showing the effect of parameters in tuning the time period of the oscillations is given.
We showed that the fixed frequency, variable amplitude profile can arise from a saturation function embedded in a negative feedback loop with linear dynamics. 
We derived explicit values of the amplitude and frequency as well as the feedback strength required for this behaviour. 
We showed how both oscillator types could be unified in a single block diagram perspective. 
These results are a step towards developing insights in realistic biological oscillators.

\section{Methods}
\subsection{Block Diagram}
A control-theoretic framework is formulated to investigate the amplitude–frequency characteristics of biological oscillators. 
The framework incorporates a nonlinear saturation function and linear dynamical elements interconnected through positive and negative feedback loops. 
The feedback gains and saturation parameters were treated as design parameters governing the closed-loop dynamics.

\subsection{Analysis of Amplitude–Frequency profiles}
The two amplitude–frequency profiles were analyzed by considering alternative arrangements of the saturation nonlinearity, linear dynamics, and feedback loops. 
For the fixed amplitude, variable frequency profile, the saturation function was incorporated into an inner positive feedback loop within an outer negative feedback loop containing linear dynamics. 
The outputs of the saturation block wrapped with a positive feedback and the linear dynamics block were analyzed to derive conditions on the feedback gains and saturation parameters and to obtain expressions for the oscillation amplitude and frequency.
For the fixed frequency, variable amplitude profile, the saturation function was incorporated into a negative feedback loop with linear dynamics. 
The corresponding closed-loop relationships were examined to determine the conditions on the feedback strength and system parameters that produce the desired amplitude–frequency behaviour. 
In both cases, the analysis focused on identifying the parameter conditions associated with the respective profiles.

\section{Results}
\subsection{Fixed Amplitude, Variable Frequency}
\begin{figure}[H]
    \begin{center}
    \includegraphics[scale=0.4]{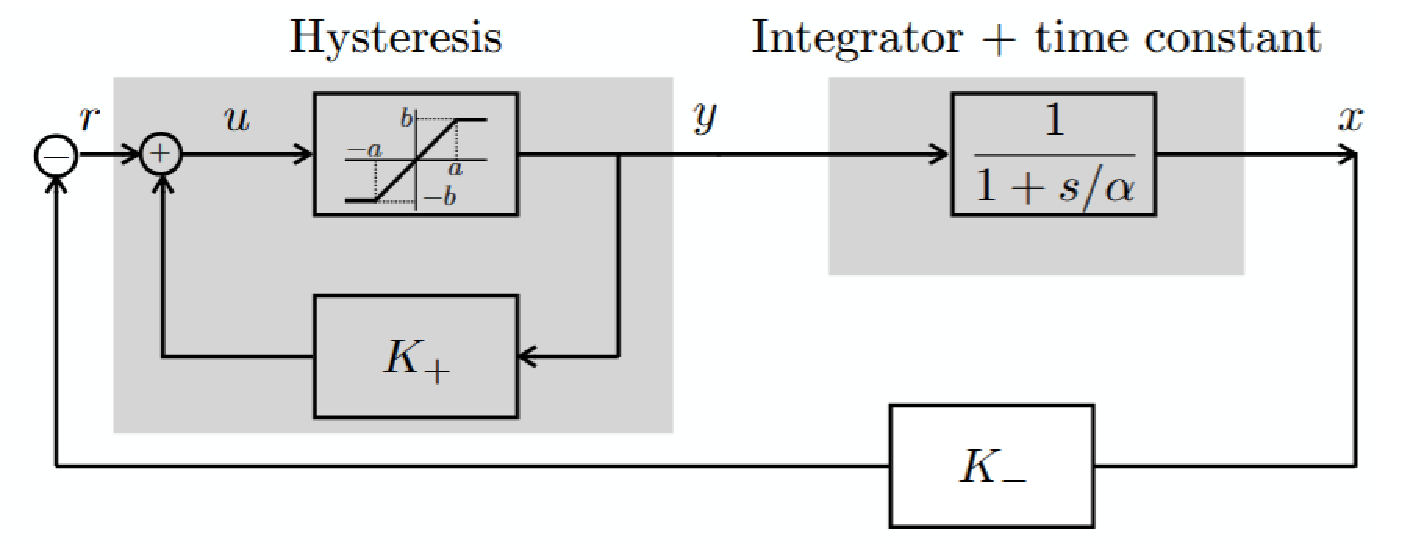}
    \caption{Block diagram representation of a fixed amplitude, variable frequency oscillator.}
    \label{Fig2}
    \end{center}
\end{figure}
We considered a quantitative generalization of the block diagram shown in the previous section (Figure \ref{Fig2}). 
The output oscillation waveform is $y(t)$. 
The saturation function had a linear behaviour when the input $u \in [-a, a] $ and saturated at $\pm b$,
    \begin{align}
        y =
            \begin{cases}
                -b, &\text{if } u < -a\\
                \dfrac{b}{a}u, &\text{if } -a \leq u < a\\
                b, &\text{if } u \geq a.
            \end{cases}
    \end{align}
Such saturation functions are almost universal features in biological systems, arising directly from resource constraints. 
This saturation function is wrapped inside a positive feedback loop with gain $K_+$. 
This entire block is wrapped in an outer negative feedback loop with a block $G(s) =
\dfrac{1}{1 + \dfrac{s}{\alpha}}$ and feedback gain $K_-$. 
This is a more realistic form of the integrator dynamics considered in Figure 1. 
The output of $G(s)$ is denoted by $x$.

\begin{prop}
    If $K_+ > \dfrac{a}{b}$ and $K_- > K_+ - \dfrac{a}{b}$ in the block diagram shown in Figure \ref{Fig2}, the amplitude-frequency characteristic is horizontal with amplitude $b$ and frequency tunable by $\alpha, ~a, ~b, ~K_+$ and $K_-$.
\end{prop}

\begin{proof}
    The inner positive feedback loop exhibits hysteresis if $K_+ > \dfrac{a}{b}$ (\cite{kumar2026quantitative}, see also \cite{astrom2008feedback, sepulchre2019control}). 
    We repeat the core idea here for completeness.
    Suppose the output of the saturation block $y = b$. 
    Then $u = r + K_+ y = r + K_+ b > a \Rightarrow r > a - K_+ b$. 
    Similarly, if $y = - b \Rightarrow r < - a + K_+ b$. 
    The $r-axis$ and $y-axis$ limits of the hysteresis are $\pm (a - K_+ b)$ and $\pm b$, respectively. 
    Suppose the output $y$ of the hysteresis block is $b$. Then, the variable $x$ at the output of the first-order transfer function is
    \begin{equation}\label{TypeI_plant_a}
        x(t) = b\big(1-e^{-\alpha t}\big) + x(0) e^{-\alpha t} ,
    \end{equation} 
    where $x(0)$ is the initial condition at time $t = 0$. 
    Due to the negative feedback, the input to the hysteresis block is $r = -K_-x$. 
    This is able to trigger the transition $y = b \rightarrow -b$ if $- K_- x = a - K_+ b$. 
    Since the maximum value of $x$ is $b$, this condition becomes $- K_- b < a - K_+ b \Rightarrow K_- > K_+ - a b$ .
    The value of $x$ at the time $t_1$ of this transition $y = b \rightarrow -b$ is $x(t_1) = \dfrac{K_+ b - a}{K_-}$. 
    Proceeding along similar lines, when the output of the hysteresis block $y = -b$, the variable $x$ after the first-order transfer function block is
    \begin{equation}\label{TypeI_plant_b}
        x(t) = b\big(1-e^{-\alpha (t-t_1)}\big) + x(0) e^{-\alpha (t-t_1)} .
    \end{equation}
    This triggers the transition $y(t) = - b \rightarrow b$ if $-K_- x = - a + K_+ b$. 
    Since the minimum value of $x$ is $-b$, this condition becomes $K_- b > - a + K_+ b \Rightarrow K_- > K_+ - a b$ . 
    The value $x$ at the time of the transition $y = - b \rightarrow b$ is $\dfrac{-(K_+ b - a)}{K_-}$.
    As this completes the cycle, this value is the same as $x(0)$.
    The time-period $T$ can be obtained by solving
    \begin{align*} \label{time_period}
        \dfrac{K_+ b - a}{K_-} &= b\Big(1-e^{-\alpha T/2}\Big) + \dfrac{-(K_+ b - a)}{K_-} e^{-\alpha T/2},\\
        \implies T &= \dfrac{2}{\alpha} \ln \Bigg\{ \dfrac{K_- + K_+ - \dfrac{a}{b}}{K_- - \Big( K_+ - \dfrac{a}{b}\Big)} \Bigg\}.  \numberthis
    \end{align*}
    The frequency is $\dfrac{1}{T}$ . 
    Therefore, the amplitude of the output $y$ is fixed at $b$ and the frequency is tunable by the other parameters.
\end{proof}

\subsection{Fixed Frequency, Variable Amplitude}
\begin{figure}[H]
    \begin{center}
    \includegraphics[scale=0.4]{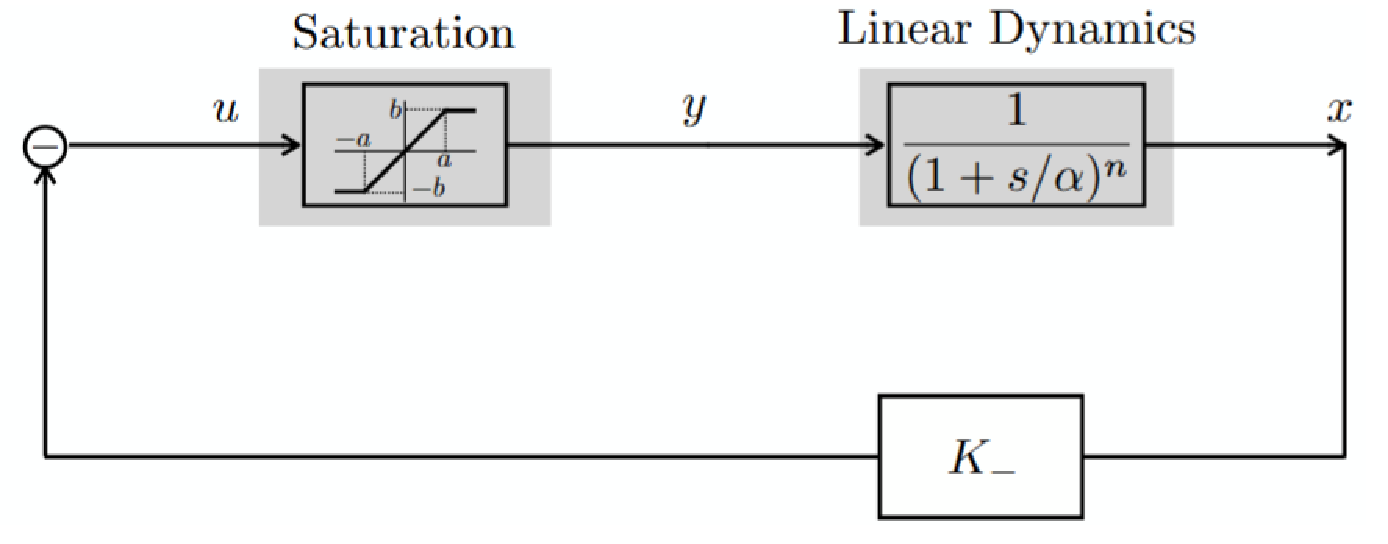}
    \caption{Block diagram representation of a fixed frequency, variable amplitude oscillator.}
    \label{Fig4}
    \end{center}
\end{figure}
We next considered a variation of the block diagram with a shift in the emphasis from the nonlinear block to the linear dynamics block (Figure \ref{Fig4}). 
Specifically, we considered the nonlinear block to be a saturation function and the linear dynamics block to be of the form $G(s) = \dfrac{1}{\bigg(1 + \dfrac{s}{\alpha} \bigg)^n}$ . 
The particular form of the linear dynamics block was chosen as the simplest representation that paralleled the structure of biological oscillators.

\begin{prop}
    If the root locus $1 + KG(s) = 0$ intersects the imaginary axis at $\pm j \omega_0$ for $K = K_0$ with all other roots in the left-half of the complex plane, and $K_- \dfrac{b}{a} = K_0$, then the amplitude-frequency profile in the block diagram in Figure \ref{Fig4} has a vertical part with the frequency $\omega_0$ and an amplitude in the range $[0, b)$.
\end{prop}

\begin{proof}
    The proof relies on the circulation of a test sinusoidal signal through the loop. 
    Consider a sinusoidal signal $y(t) = A \sin \omega_0 t$ injected into the linear dynamics block. 
    The output of the linear dynamic block has magnitude $A\mid G(j\omega_0)\mid$ and phase $\angle G(j\omega_0)$. 
    Because $1+K_0 G(j\omega_0) = 0$, $\mid G(j\omega_0)\mid = \dfrac{1}{K_0}$ and $\angle G(j\omega_0) = \pi$. 
    After the negative feedback, the input at the saturation block is a sinusoid with magnitude $\dfrac{AK_-}{K_0} = \dfrac{Aa}{b}$, frequency $\omega_0$, and phase $2 \pi$.
    As $A < b$, the amplitude remains in the linear region of the saturation function and the output of the saturation block is $A sin \omega_0 t$, which is the same as the signal injected into the linear dynamics block. 
    Therefore, any sinusoidal signal $y(t) = A sin \omega_0t$ with $A \in [0, b)$ is a solution.
    
    This has a fixed frequency and a variable amplitude. 
    For the saturation block to act as a linear gain block, we assume the amplitude of the oscillations multiplied by the gain $K_-$ in the feedback path should be less than $a$.
\end{proof}

\begin{rem}
    The injection of the sinusoidal signal $y(t) = A \sin \omega_0 t$ into the linear dynamics block produces a forced periodic response with a frequency ($\omega_0$). 
    The output represents sustained oscillatory behavior but does not constitute a limit cycle, since it is externally imposed.
    The pair of imaginary-axis roots at ($\pm j\omega_0$) give rise to sinusoidal oscillations at ($\omega_0$), with the amplitude determined by the initial conditions.
    Thus, the injection of sinusoidal signal can guarantee an oscillatory behaviour, but it does not establish the existence of a limit cycle.
\end{rem}

\begin{rem}
    For the specific form of transfer function $G(s)$ considered, the frequency $\omega_0$ can be explicitly calculated.
    The value $(\omega_0, K_0)$ can be obtained from the solution of the $1+K_0 G(j\omega_0) = 0$.   
    \begin{align*}
        &1+K_0 G(j\omega_0) = 0, \\
        \implies& 1 + \dfrac{K_0}{\bigg(1 + j\dfrac{\omega_0}{\alpha} \bigg)^n} = 0, \\
        \implies& \bigg(1 + j\dfrac{\omega_0}{\alpha} \bigg)^n = - K_0, \\
        \implies& 1 + j\dfrac{\omega_0}{\alpha} = K_0^{1/n} (-1)^{1/n}, \\
        \implies& 1 + j\dfrac{\omega_0}{\alpha} = K_0^{1/n} e^{\phi_m}, ~~m = 0,1,\ldots, n-1, \numberthis \label{rootlocus}
    \end{align*}
    where $\phi_m = \dfrac{(2m + 1)}{n} \pi$.
    We separate \eqref{rootlocus} into its real and imaginary components:
    \begin{align}
        1 &= K_0^{1/n} \cos\phi_m, \label{eq:real_comp} \\
        \omega_0 &= \alpha K_0^{1/n} \sin\phi_m. \label{eq:im_comp}
    \end{align}
    \begin{figure}[H]
        \begin{center}
        \includegraphics[trim={1.5cm 0 1.5cm 0}, clip,scale=0.9]{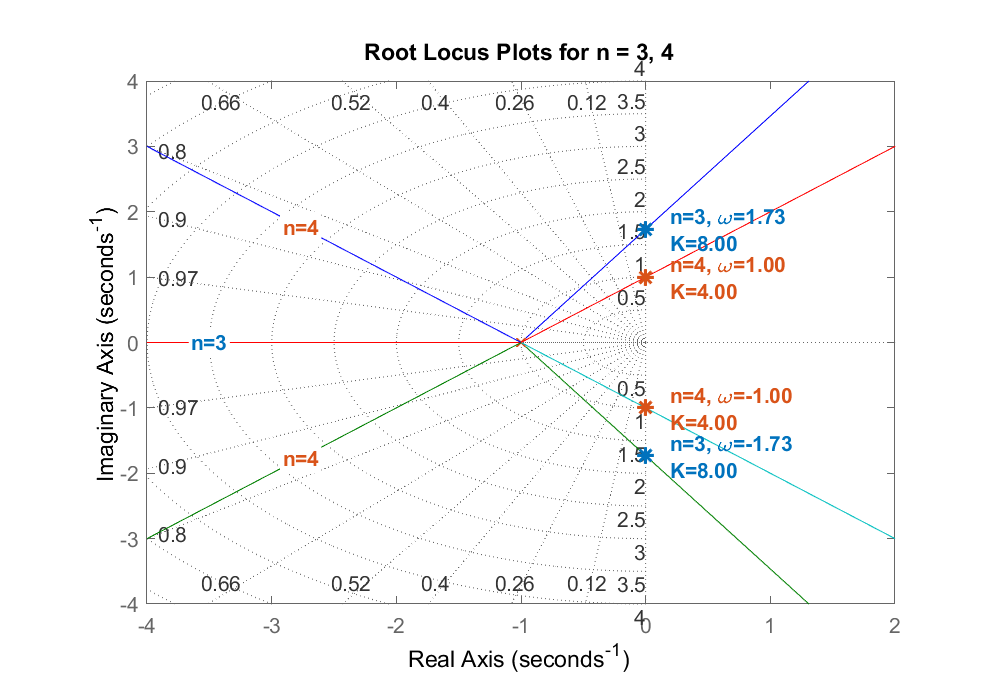}
    \caption{Root locus of the block diagram showing the values of $\omega$ and the gain where it crosses the imaginary axis for $n = 3, ~4$ and $\alpha = 1$.}
        \label{Rlocus}
        \end{center}
    \end{figure}
    \begin{figure}[H]
        \begin{center}
        \includegraphics[trim={1.5cm 0 1.5cm 0}, clip,scale=0.9]{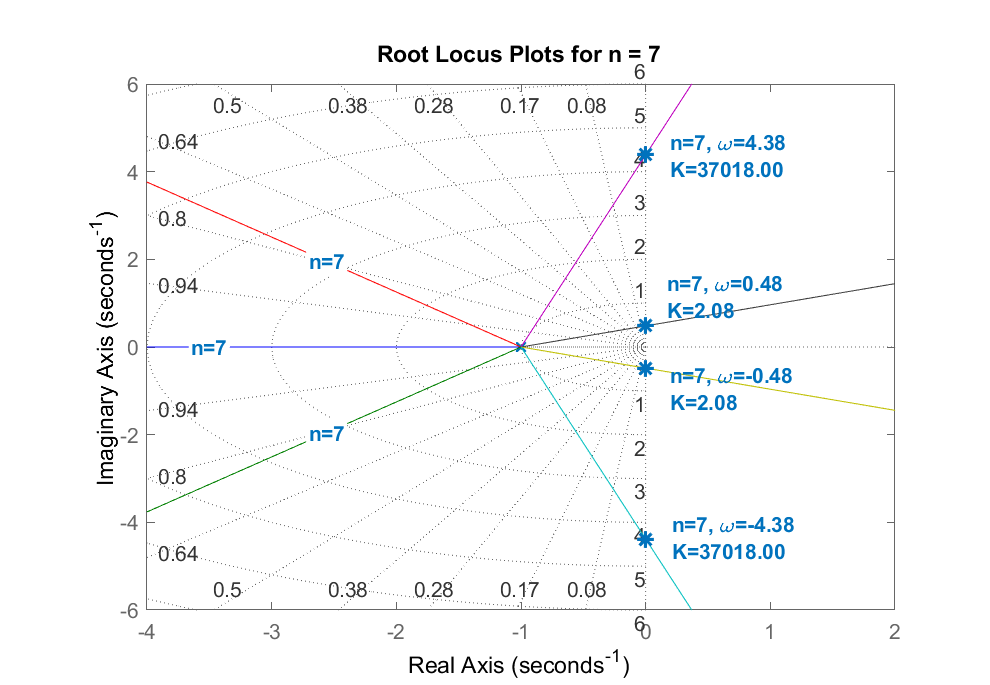}
        \caption{Root locus of the block diagram showing the values of $\omega$ and the gain where it crosses the imaginary axis for $n = 7$ and $\alpha = 1$. There are multiple crossings of the imaginary axis.}
        \label{Rlocus_n7}
        \end{center}
    \end{figure}
    The value $(\omega_0, K_0)$ for different values of $n$ are as follows:
    \begin{itemize}[leftmargin=*]
        \item \textbf{For $n = 1$:} The phase angle is $\phi_0 = \pi$, which yields $1 = K_0 \cos(\pi) = -K_0$, strictly requiring $K_0 = -1$. Since $K_0$ comprises of positive constants and must be positive, no solution exists.
        \item \textbf{For $n = 2$:} The phase angles are $\phi_0 = \pi/2$ and $\phi_1 = 3\pi/2$. Substituting these into Eq.~(\ref{eq:real_comp}),  we have $K_0^{1/2} \cos(\pi/2) = 0$ and $K_0^{1/2} \cos (3\pi/2) = 0$. It creates the contradiction $1 = 0$. Thus, no solution exists for $n \le 2$.
        \item \textbf{For $n \geq 3$:} If $\cos \phi_m \neq 0$, $K_0 = \sec^n \phi_m$. Due to the shape of the secant function, the first crossing happens when $m = 0$, $K_0 = \sec^n \phi_0$ and $\omega_0 = \alpha \tan \phi_0$. Example root locus diagrams were plotted for $n = 3, 4$ (Figure \ref{Rlocus}) and $n = 7$ (Figure \ref{Rlocus_n7}).
    \end{itemize}
    Apart from the mathematical contradictions, the absence of oscillatory behavior for $n \leq 2$ can be directly explained using the Barkhausen stability criterion \cite{razavi2001design}. 
    For the closed-loop system shown in Figure \ref{Fig4} to show sustained oscillations, the total phase shift around the loop must equal a multiple of $360^{\circ}$. 
    Because the negative feedback inherently introduces a phase shift of $180^{\circ}$, the linear dynamics block must provide at least an additional $180^{\circ}$ of phase shift. 
    For $n=1$, the maximum possible phase shift is $90^{\circ}$, strictly ruling out any sustained oscillatory behavior. 
    For $n=2$, a total phase shift of $180^{\circ}$ is theoretically possible and requires each of the two stages to contribute $90^{\circ}$. 
    Mathematically, the phase shift can only be achieved as the frequency $\omega_0$ in \eqref{eq:im_comp} tends to infinity ($\omega_0 \to \infty$). 
    Because an infinite oscillation frequency is physically infeasible in realistic models, no oscillatory solutions exist for $n \leq 2$.
\end{rem}

\subsection{General Representation}
\begin{figure}[H]
    \begin{center}
    \includegraphics[scale=0.4]{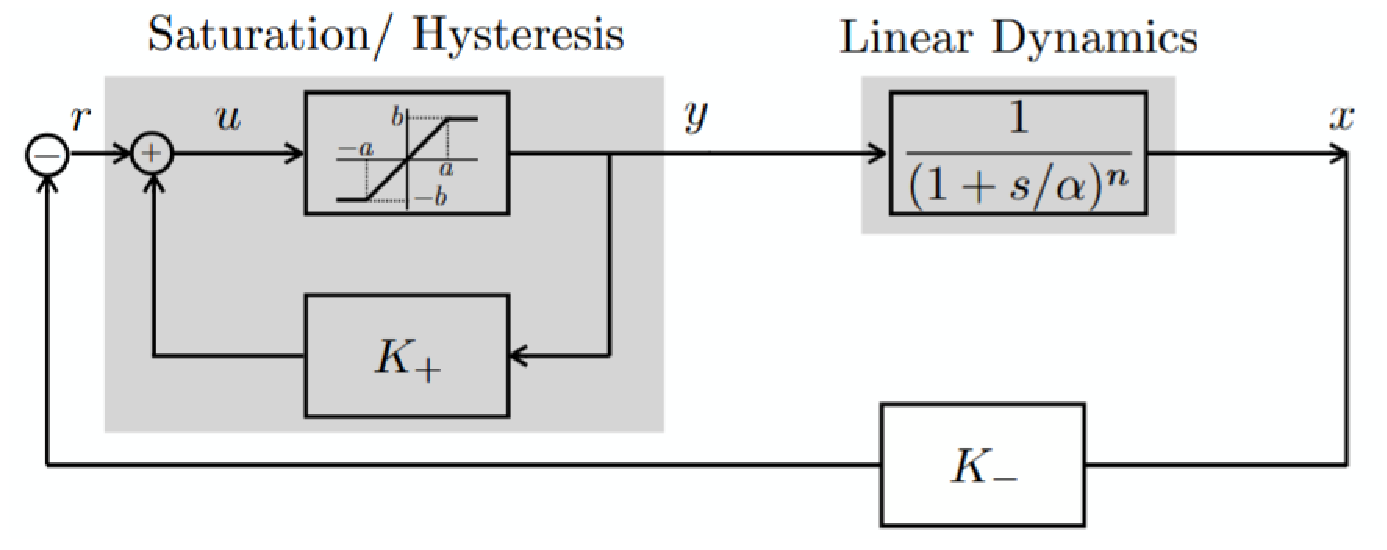}
    \caption{Unified block diagram representation of both oscillator types.}
    \label{Fig7}
    \end{center}
\end{figure}
Based on the above results, we propose a general block diagram representation that includes oscillators from both types (Figure \ref{Fig7}). 
Example of a parameter set that can give square wave oscillations with fixed amplitude and variable frequency are $a = b = 1, ~K_+ = 2, ~\alpha = 1, ~n = 1, ~K_- = 2$ (simulation results in Figure \ref{Type_I}). 
The output of the saturation block with positive feedback (hysteresis) is a square wave with the amplitude equal to $b=1$. 
This is shown in the hysteresis output ($y$) in Figure \ref{Type_I}. 
The plant output ($x$) shown in Figure \ref{Type_I} follows the expression given in \eqref{TypeI_plant_a} and \eqref{TypeI_plant_b}.
\begin{figure*}
    \begin{center}
    \includegraphics[trim={2cm 0 2cm 0}, clip,scale=0.75]{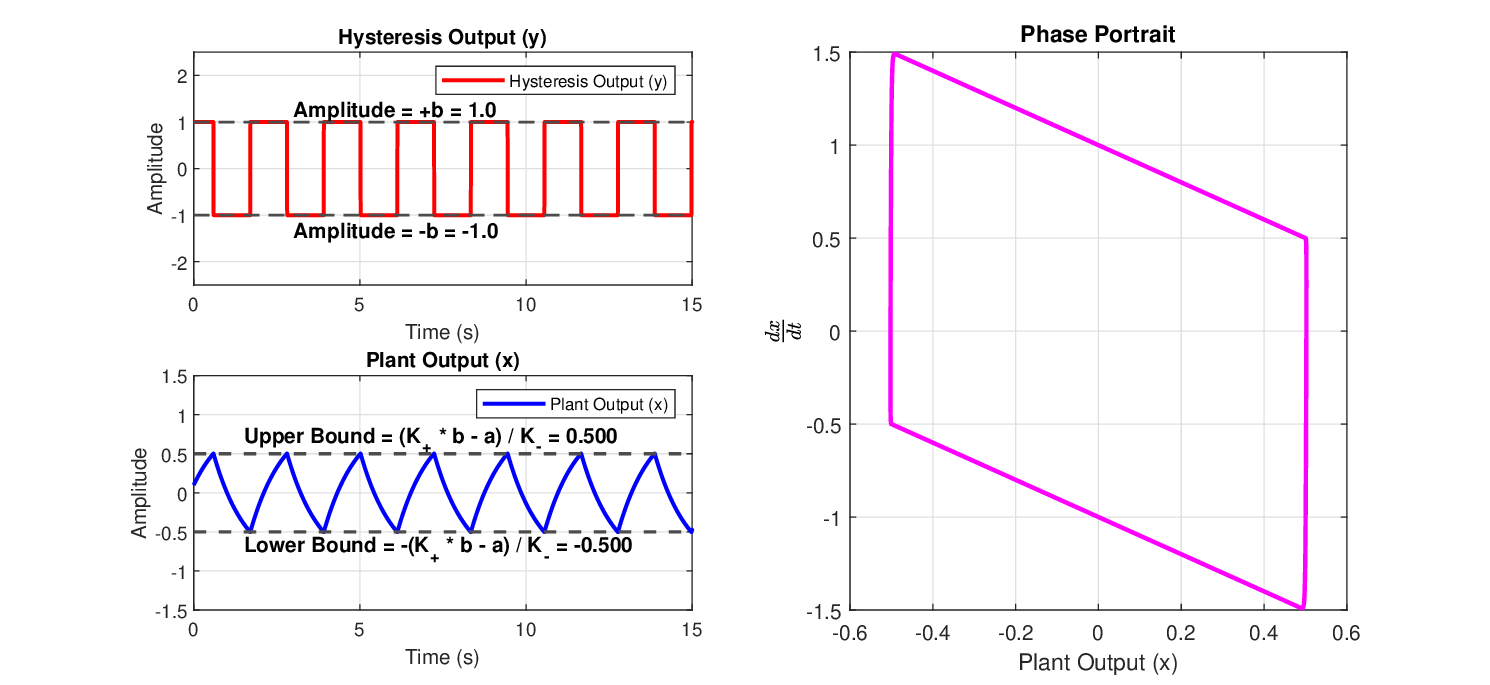}
    \caption{Simulation results of the general representation in Figure \ref{Fig7} for the parameter set $a = b = 1, ~K_+ = 2, ~\alpha = 1, ~n = 1, ~K_- = 2$.}
    \label{Type_I}
    \end{center}
\end{figure*}
\begin{figure*}
    \begin{center}
    \includegraphics[trim={2cm 0 2cm 0}, clip,scale=0.75]{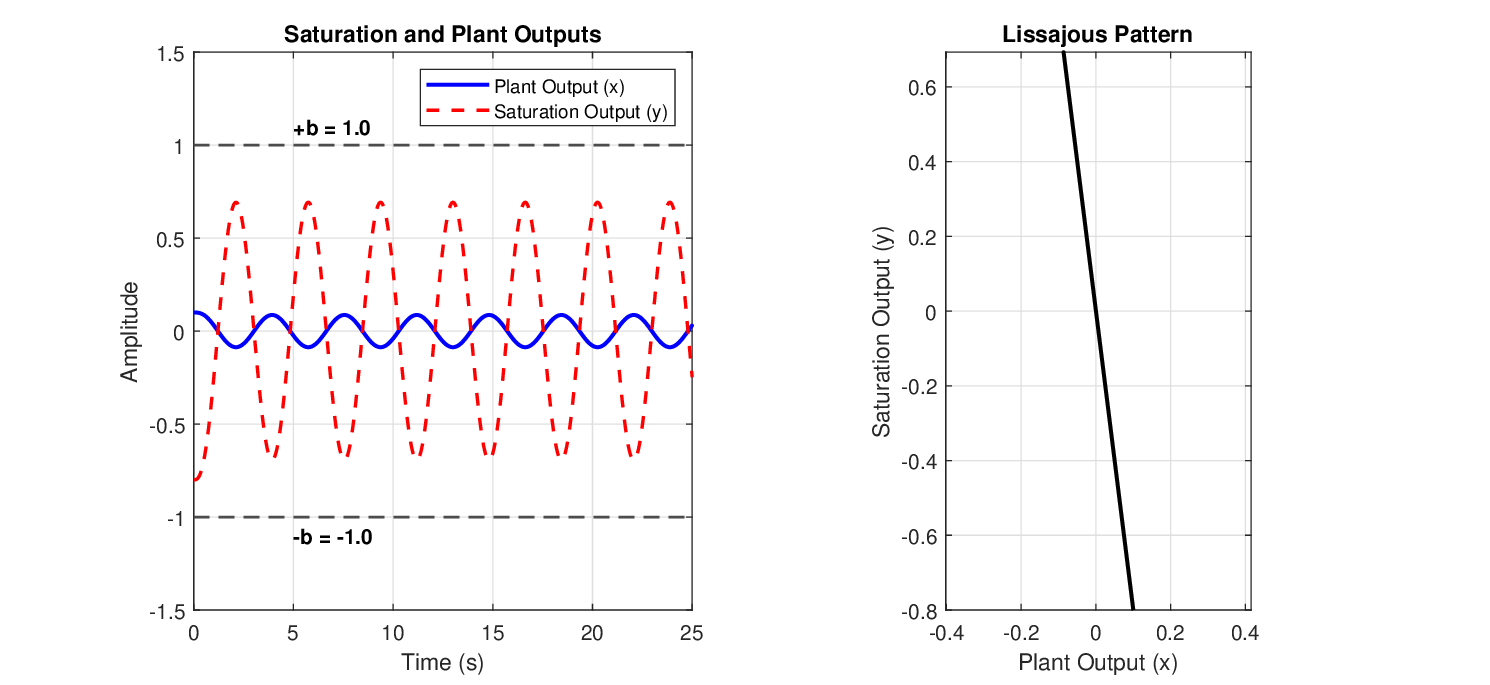}
    \caption{Simulation results of the general representation in Figure \ref{Fig7} for the parameter set $a = b = 1, ~K_+ = 0, ~\alpha = 1, ~n = 3, ~K_- = 8$.}
    \label{Type_II}
    \end{center}
\end{figure*}
Example of a parameter set that can give sinusoidal oscillations with fixed frequency and variable amplitude are $a = b = 1, ~K_+ = 0, ~\alpha = 1, ~n = 3, ~K_- = 8$ (simulation results in Figure \ref{Type_II}). 
The sinusoidal outputs of the saturation block and the linear dynamics block (Plant) are given in Figure \ref{Type_II} (Saturation and Plant Outputs). 
The amplitudes of the outputs are within the range $[0, b)$, where $b=1$. 
The Lissajous pattern between the the two sinusoidal outputs is a straight line. 
This proves the frequency is same for both the outputs \cite{al2012geometrical}, giving sinusoidal oscillations with fixed frequency.
The negative slope of the straight line shows the two signals are $180^{\circ}$ apart which could be accounted to the negative feedback in the block diagram in Figure \ref{Fig4}.

\section{Discussion}
The block diagram representation in Figure \ref{Fig2} provides a clear delineation of the minimum positive and negative feedback strengths needed for oscillations. 
These conditions were expressed in terms of the parameters of the saturation function and the dynamics.
In particular, the relationship between the positive and negative feedback strengths determines the regions of the parameter space, as illustrated in Figure \ref{Bifurcation}.

\begin{figure}[H]
    \begin{center}
    \includegraphics[trim={1cm 0 1.3cm 0}, clip, scale=0.75]{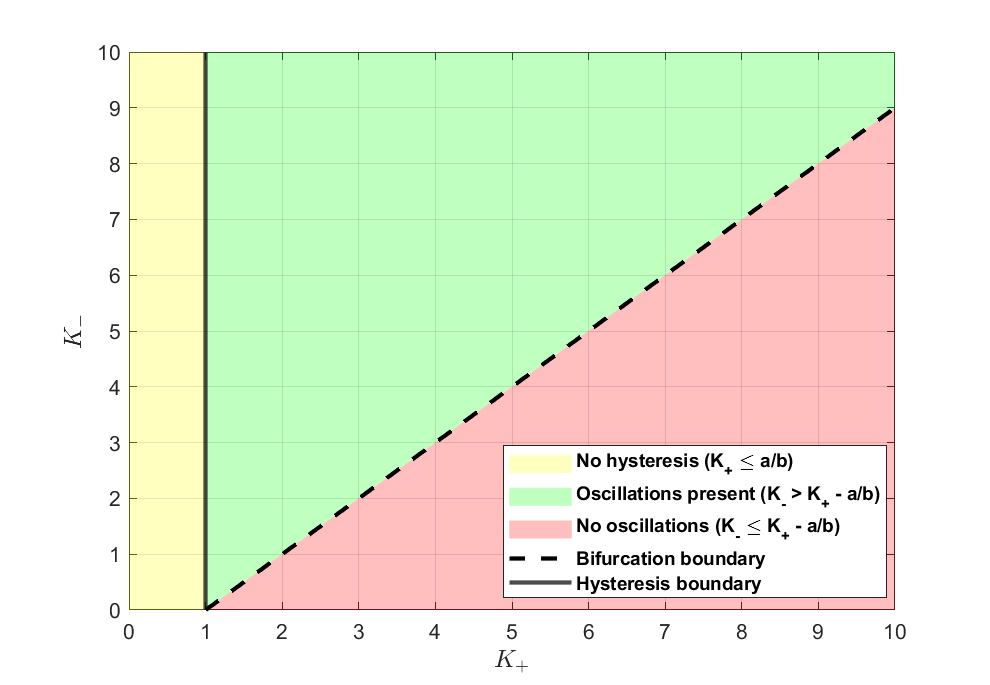}
    \caption{Classification of the parameter space based on the relation between the positive and negative feedback strength for $a = b = 1$.}
    \label{Bifurcation}
    \end{center}
\end{figure}

The sensitivity function $S_p^T = \frac{p}{T}\frac{\partial T}{\partial p}$, which quantifies the change in $T$ resulting from a change in parameter $p$ \cite{nise2000control, ogata2002modern}, is utilized to study the impact of the parameters $\alpha, ~a, ~b, ~K_+$ and $K_-$ on the tuning of the time-period $T$ in \eqref{time_period}.
Differentiating \eqref{time_period} gives the following sensitivity expressions of $T$ for each parameter:
\begin{align}
    S_\alpha^T &= -1, \label{eq:sens_alpha} \\ 
    S_a^T &= -\frac{2 \left(\frac{a}{b}\right) K_-}{\left[ K_-^2 - \left(K_+ - \frac{a}{b}\right)^2 \right] D}, \label{eq:sens_a} \\
    S_b^T &= \frac{2 \left(\frac{a}{b}\right) K_-}{\left[ K_-^2 - \left(K_+ - \frac{a}{b}\right)^2 \right] D} = -S_a^T, \label{eq:sens_b} \\
    S_{K_+}^T &= \frac{2 K_+ K_-}{\left[ K_-^2 - \left(K_+ - \frac{a}{b}\right)^2 \right] D}, \label{eq:sens_kplus} \\
    S_{K_-}^T &= -\frac{2 \left(K_+ - \frac{a}{b}\right) K_-}{\left[ K_-^2 - \left(K_+ - \frac{a}{b}\right)^2 \right] D}, \label{eq:sens_kminus}   
\end{align}
where $D = \ln \left\{ \frac{K_- + K_+ - \frac{a}{b}}{K_- - \left( K_+ - \frac{a}{b}\right)} \right\}$.
\begin{figure*}
    \begin{center}
        \includegraphics[trim={5.25cm 0 5cm 0}, clip, scale=0.43]{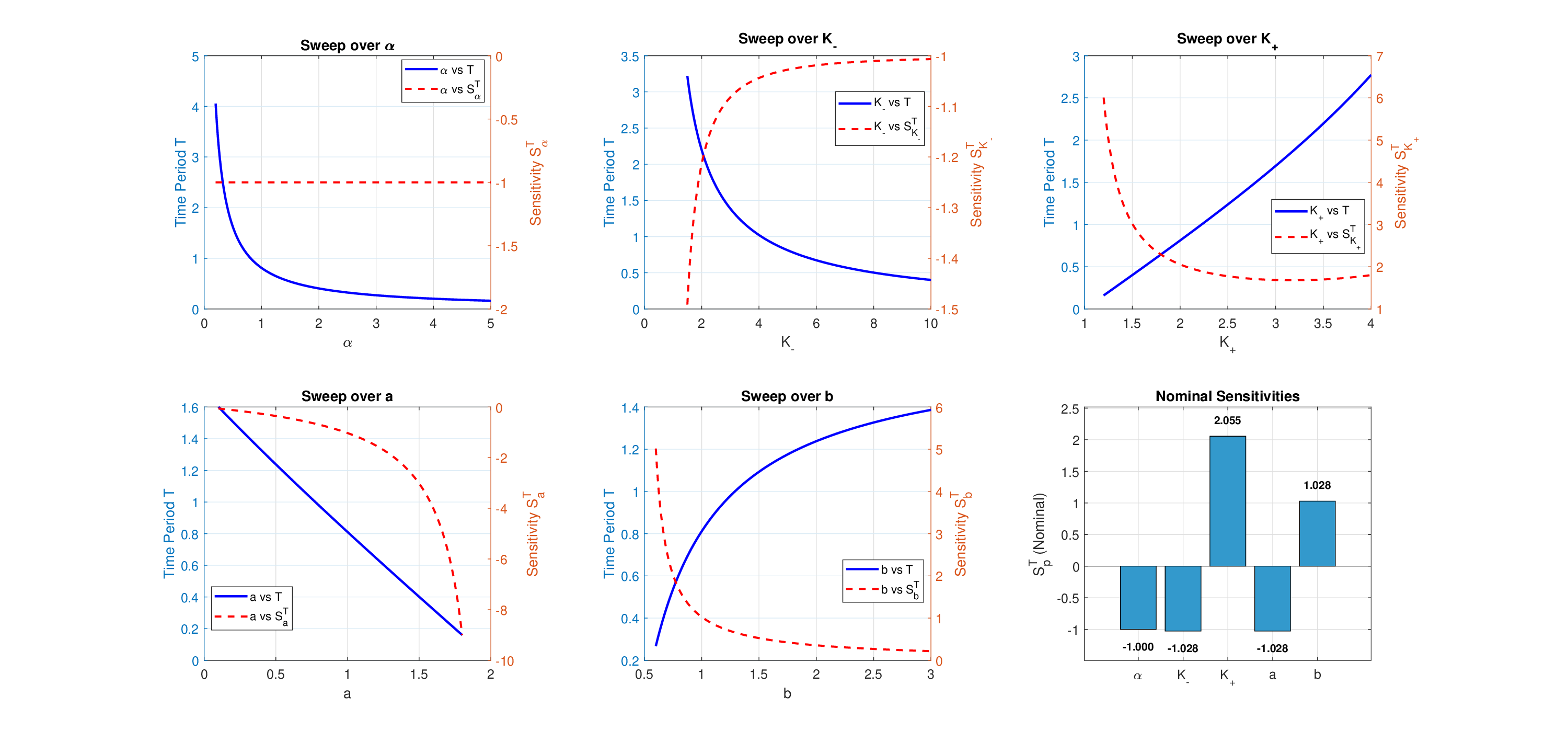}
        \caption{Dual-axis parameter sensitivity sweeps mapping the time-period against sensitivity, alongside a bar chart of nominal sensitivities. The nominal parameter set used for the simulation are as follows: $\alpha = 1, ~a = 1, ~b = 1, ~K_+ = 2$ and $K_- = 5$.}
        \label{fig:sens_sweeps}
    \end{center}
\end{figure*}
The sensitivity expressions provide insights into the impact of each parameter in tuning the time-period of the oscillations. 
\eqref{eq:sens_alpha} indicates an inverse proportionality between $\alpha$ and $T$, while \eqref{eq:sens_a} and \eqref{eq:sens_b} show that $a$ and $b$ influence $T$ oppositely through the ratio $a/b$. For the parameter regime considered, increase in $K_+$ and $b$ increase the time period, whereas $\alpha$, $a$, and $K_-$ decrease it.
This is consistent with the plots in Figure~\ref{fig:sens_sweeps}. 
At the nominal parameter set, the largest sensitivity is associated with $K_+$ while $|S_\alpha^T|=1$. 
Thus, $K_+$ provides the strongest relative tuning of $T$. 
The sweeps further show that the sensitivity to $a$ increases in magnitude at higher parameter values, whereas the sensitivity to $K_+$ is largest at lower values, indicating that $a$ and $K_+$ provide stronger control over the time period in these respective parameter regimes. 

The block diagram in Figure \ref{Fig4} can be further analyzed using the methods of describing functions and harmonic balance \cite{slotine1991applied, sastry1999nonlinear}. 
The main idea is to replace the nonlinear block with the first harmonic component, say $N (A, \omega)$. 
This ``transfer function” $N (A, \omega)$ is usually dependent on the amplitude $A$ of the input signal, unlike in the purely linear analysis. 
The solution of the resulting harmonic balance equations estimates the amplitude and frequency of possible oscillations. 
While this method is useful in analyzing such systems, it is only approximate.
Several other theoretical frameworks have also been used to analyze the block diagram model shown in Figure \ref{Fig4}.
However, the emphasis has primarily been on the stability of the equilibrium point or on the existence of oscillations. 
Examples include the use of small-gain theorem on monotone systems used in systems biology \cite{angeli2008oscillations}. 
More recently, classical nonlinear control topics such as the Lur’e problem and the circle criterion have been generalized for the analysis of oscillations using dominance theory \cite{miranda2018analysis, che2023dominant, che2021shaping, che2021tunable, che2022analog}.

The linear dynamics blocks considered in Figures \ref{Fig2}, \ref{Fig4}, and \ref{Fig7} can be generalized to an arbitrary linear dynamical block $G(s)$. 
This connects the proposed framework to the broader class of relay-feedback systems for which established methods are available to characterize the existence, stability, amplitude, and frequency of periodic solutions \cite{aastrom1995oscillations, johansson2002limit, gonccalves2003regions}. 
However, the present work focuses on understanding the rationale behind the proposed amplitude–frequency design principle. 
The selected linear dynamics block allows the study of the design principle without introducing the additional analytical complexity required for a general $G(s)$. 

\section{Conclusion}
We found a useful connection between a design principle for biological oscillators and their block diagram models.
This connection allowed an explicit characterization of the amplitude-frequency profile of biological oscillators. 
While such models are idealizations of more realistic models from which the design principle was obtained via numerical simulations, the primary advantage of the block diagram models is in providing a theoretical basis for the design principle. 
Further, the connection provided a succinct representation of the role of positive and negative feedback in the overall amplitude-frequency profile. 

These results should help develop systems-level insight in more realistic models of biological oscillators.
Mapping these distinct block diagram configurations back to biological networks will provide a powerful quantitative tool for both analyzing natural and synthetic biological oscillators. 
The division of oscillators into fixed amplitude/variable frequency and fixed frequency/variable amplitude architectures reveals that a biological oscillator depends fundamentally on the topology of its feedback loops. 
Future works can use this mapping by combining positive and negative feedback, specifically tuning the relative strength of the positive and the negative loops to build oscillators capable of operating over wide frequency and amplitude ranges. 
Treating complex biochemical reactions as transfer functions with nonlinearities will allow us to design custom oscillators that can maintain predefined amplitude and frequency profiles across diverse physiological environments.

\bibliographystyle{ieeetr}
\bibliography{Reference}

\end{document}